\documentclass[twoside,11pt]{article}

\usepackage{times}
\usepackage{fullpage}

\usepackage{graphicx} 
\usepackage{subfigure}

\usepackage[round]{natbib}

\usepackage{algorithm}
\usepackage{algorithmic}
\usepackage{hyperref}

\usepackage{multirow}
\usepackage{enumitem}
\usepackage{xspace}
\usepackage{amssymb,amsfonts,latexsym,amsmath,amsthm}

\theoremstyle{plain}

\newtheorem{lemma}{Lemma}

\newtheorem{claim}{Claim}

\usepackage{amsmath,amsfonts}
\usepackage{amsthm,amssymb,amsopn}
\usepackage{bm} 

\newcommand{\vct}[1]{\boldsymbol{#1}} 
\newcommand{\mat}[1]{\boldsymbol{#1}} 

\newcommand{\T}{^{\textrm T}} 

\DeclareMathOperator{\argmax}{arg\,max}
\DeclareMathOperator{\argmin}{arg\,min}

\newtheorem{thm}{Theorem}
\newtheorem{assumption}{Assumption}

\newcommand{\eat}[1]{}

\newcommand{\innerp}[1]{\left\langle #1 \right\rangle}
\DeclareMathOperator{\supp}{supp}
\DeclareMathOperator{\sgn}{sgn}
\DeclareMathOperator{\diam}{diam}

\usepackage[title]{appendix}

\begin{document}

\title{A Distributed Frank-Wolfe Algorithm for Communication-Efficient Sparse Learning}
\author{Aur\'elien Bellet\thanks{LTCI UMR 5141, T\'el\'ecom ParisTech \& CNRS, \texttt{aurelien.bellet@telecom-paristech.fr}. Most of the work in this paper was carried out while the author was affiliated with the Department of Computer Science of the University of Southern California.}
\and 
Yingyu Liang\thanks{Department of Computer Science, Princeton University,
\texttt{yingyul@cs.princeton.edu}.}
\and
Alireza Bagheri Garakani\thanks{Department of Computer Science, University of Southern California,
\texttt{\{bagherig,feisha\}@usc.edu}.}
\and
Maria-Florina Balcan\thanks{School of Computer Science, Carnegie Mellon University, \texttt{ninamf@cs.cmu.edu}.}
\and
Fei Sha\footnotemark[3]
}

\date{}

\maketitle

\begin{abstract}
Learning sparse combinations is a frequent theme in machine learning. In this paper, we study its associated optimization problem in the distributed setting where the elements to be combined are not centrally located but spread over a network. We address the key challenges of balancing communication costs and optimization errors. To this end, we propose a distributed Frank-Wolfe (dFW) algorithm. We obtain theoretical guarantees on the optimization error $\epsilon$ and communication cost that do not depend on the total number of combining elements. We further show that the communication cost of dFW is optimal by deriving a lower-bound on the communication cost required to construct an $\epsilon$-approximate solution. We validate our theoretical analysis with empirical studies on synthetic and real-world data, which demonstrate that dFW outperforms both baselines and competing methods. We also study the performance of dFW when the conditions of our analysis are relaxed, and show that dFW is fairly robust.
\end{abstract}


\section{Introduction}

Most machine learning algorithms are designed to be executed on a single computer that has access to the full training set. However, in recent years, large-scale data has become increasingly distributed across several machines ({\it nodes}) connected through a network, for storage purposes or because it is collected at different physical locations, as in cloud computing, mobile devices or wireless sensor networks \citep{Lewis2004,Armbrust2010}. In these situations, communication between machines is often the main bottleneck and the naive approach of sending the entire dataset to a coordinator node is impractical.
This distributed setting triggers many interesting research questions. From a fundamental perspective, studying the tradeoff between communication complexity \citep{Yao1979,abelson1980lower} and learning/optimization error has recently attracted some interest \citep{Balcan2012,DaumeIII2012,Shamir2014}. From a more practical view, a lot of research has gone into developing scalable algorithms with small communication and synchronization overhead \citep[see for instance][]{Lazarevic2002,Forero2010,Boyd2011,Duchi2012,Dekel2012,Balcan2013,Yang2013} and references therein).

In this paper, we present both theoretical and practical results for the problem of learning a sparse combination of elements ({\it atoms}) that are spread over a network. This has numerous applications in machine learning, such as LASSO regression \citep{Tibshirani1996}, support vector machines \citep{Cortes1995}, multiple kernel learning \citep{Bach2004}, $\ell_1$-Adaboost \citep{Shen2010} and sparse coding \citep{Lee2006}. Formally, let $\mat{A} = [\vct{a}_1\dots\vct{a}_n]\in\mathbb{R}^{d\times n}$ be a finite set of $n$ $d$-dimensional atoms (one per column). Here, atoms broadly refer to data points, features, dictionary elements, basis functions, etc, depending on the application. We wish to find a weight vector $\vct{\alpha}\in \mathbb{R}^n$ ($\alpha_i$ giving the weight assigned to atom $i$) that minimizes some cost function $f(\vct{\alpha}) = g(\mat{A}\vct{\alpha})$ under a sparsity constraint:\footnote{Our results also hold for sparsity constraints other than the $\ell_1$ norm, such as simplex constraints.}
\begin{equation}
\label{eq:sga}
\min_{\vct{\alpha}\in \mathbb{R}^n} \quad f(\vct{\alpha})\quad \text{s.t.} \quad \|\vct{\alpha}\|_1 \leq \beta,
\end{equation}
where the cost function $f(\vct{\alpha})$ is convex and continuously differentiable and $\beta> 0$.

We propose a distributed Frank-Wolfe (dFW) algorithm, a novel approach to solve such distributed sparse learning problems based on an adaptation of its centralized counterpart \citep{Frank1956,Clarkson2010,Jaggi2013}.
Intuitively, dFW is able to determine, in a communication-efficient way, which atoms are needed to reach an $\epsilon$-approximate solution, and only those are broadcasted through the network. We also introduce an approximate variant of dFW that can be used to balance the amount of computation across the distributed machines to reduce synchronization costs.
In terms of theoretical guarantees, we show that the total amount of communication required by dFW is upper-bounded by a quantity that does not depend on $n$ but only on $\epsilon$, $d$ and the topology of the network. We further establish that this dependency on $\epsilon$ and $d$ is worst-case optimal for the family of problems of interest by proving a matching lower-bound of $\Omega(d/\epsilon)$ communication for any deterministic algorithm to construct an $\epsilon$-approximation. 
Simple to implement, scalable and parameter-free, dFW can be applied to many distributed learning problems. We illustrate its practical performance on two learning tasks: LASSO regression with distributed features and Kernel SVM with distributed examples. Our main experimental findings are as follows: (i) dFW outperforms baseline strategies that select atoms using a local criterion, (ii) when the data and/or the solution is sparse, dFW requires less communication than ADMM, a popular competing method, and (iii) in a real-world distributed architecture, dFW achieves significant speedups over the centralized version, and is robust to random communication drops and asynchrony in the updates.

\paragraph{Outline} Section~\ref{sec:fw} reviews the Frank-Wolfe algorithm in the centralized setting. In Section~\ref{sec:dfw}, we introduce our proposed dFW algorithm and an approximate variant of practical interest, as well as examples of applications. The theoretical analysis of dFW, including our matching communication lower bound, is presented in Section~\ref{sec:theory}. Section~\ref{sec:related} reviews some related work, and experimental results are presented in Section~\ref{sec:exp}. Finally, we conclude in Section~\ref{sec:conclu}.

\paragraph{Notations}

We denote scalars by lightface letters, vectors by bold lower case letters and matrices by bold upper case letters. For $i\in\mathbb{N}$, we use $[i]$ to denote the set $\{1,\dots,i\}$. For $i\in [n]$, $\vct{e}^i\in\mathbb{R}^n$ is the all zeros vector except 1 in the $i^{\text{th}}$ entry. For a vector $\vct{\alpha}$, the support of $\vct{\alpha}\in\mathbb{R}^n$ is defined as $\supp(\vct{\alpha}) = \{i\in[n] : \alpha_i \neq 0\}$. The gradient of a differentiable function $f$ at $\vct{\alpha}$ is denoted by $\nabla f(\vct{\alpha})$.

\section{The Frank-Wolfe Algorithm}
\label{sec:fw}

%
%

\begin{algorithm}[t]
\caption{FW algorithm on general domain $\mathcal{D}$}
\label{alg:fw}
\begin{algorithmic}[1]
\STATE Let $\vct{\alpha}^{(0)} \in \mathcal{D}$
\FOR{$k = 0,1,2,\dots$}
\STATE $\vct{s}^{(k)} = \argmin_{\vct{s}\in\mathcal{D}} \innerp{\vct{s},\nabla f(\vct{\alpha}^{(k)})}$
\STATE $\vct{\alpha}^{(k+1)} = (1-\gamma) \vct{\alpha}^{(k)} + \gamma \vct{s}^{(k)}$, where $\gamma=\frac{2}{k+2}$ or obtained via line-search.
\ENDFOR
\STATE{{\bf stopping criterion}}: $\innerp{\vct{\alpha}^{(k)} - \vct{s}^{(k)},\nabla f(\vct{\alpha}^{(k)})} \leq \epsilon$
\end{algorithmic}
\end{algorithm}

\begin{algorithm}[t]
\caption{Compute $\vct{s}^{(k)}$ for $\ell_1$ / simplex}
\label{alg:local}
\begin{algorithmic}[1]
\STATE  \textbf{if} $\ell_1$ constraint
\STATE \quad $j^{(k)} = \argmax_{j\in [n]} \left|\nabla f(\vct{\alpha}^{(k)})_j\right|$
\STATE \quad $\vct{s}^{(k)} = \sgn \left[-\nabla f(\vct{\alpha}^{(k)})_{j^{(k)}}\right]\beta\vct{e}^{j^{(k)}}$
\STATE \textbf{else if} simplex constraint
\STATE \quad $j^{(k)} = \argmin_{j\in [n]} \nabla f(\vct{\alpha}^{(k)})_j$
\STATE \quad $\vct{s}^{(k)} = \vct{e}^{j^{(k)}}$
\end{algorithmic}
\end{algorithm}

The Frank-Wolfe (FW) algorithm \citep{Frank1956,Clarkson2010,Jaggi2013}, also known as Conditional Gradient \citep{Levitin1966}, is a procedure to solve convex constrained optimization problems of the form
\begin{equation}
\label{eq:genfw}
\min_{\vct{\alpha}\in \mathcal{D}} \quad f(\vct{\alpha}),
\end{equation}
where $f$ is convex and continuously differentiable, and $\mathcal{D}$ is a compact convex subset of any vector space (say $\mathbb{R}^d$). The basic FW algorithm is shown in Algorithm~\ref{alg:fw}. At each iteration, it moves towards a feasible point $\vct{s}$ that minimizes a linearization of $f$ at the current iterate. The stopping criterion uses a surrogate optimality gap that can be easily computed as a by-product of the algorithm.
Let $\vct{\alpha}^*$ be an optimal solution of \eqref{eq:genfw}. Theorem~\ref{thm:fwconv} shows that FW finds an $\epsilon$-approximate solution in $O(1/\epsilon)$ iterations.

\begin{thm}[\citealp{Jaggi2013}]
\label{thm:fwconv}
Let $C_f$ be the curvature constant of $f$.\footnote{The assumption of bounded $C_f$ is closely related to that of $L$-Lipschitz continuity of $\nabla f$ with respect to an arbitrary chosen norm $\|\cdot\|$. In particular, we have $C_f \leq L\diam_{\|\cdot\|}(\mathcal{D})^2$ \citep{Jaggi2013}.} Algorithm~\ref{alg:fw} terminates after at most $6.75C_f/\epsilon$ iterations and outputs a feasible point $\tilde{\vct{\alpha}}$ which satisfies $f(\tilde{\vct{\alpha}}) - f(\vct{\alpha}^*) \leq \epsilon$.
\end{thm}

\paragraph{Solving the subproblems} The key step in FW is to solve a linear minimization subproblem on $\mathcal{D}$ (step 3 in Algorithm~\ref{alg:fw}). When the extreme points of $\mathcal{D}$ have a special structure (e.g., when they are sparse), it can be exploited to solve this subproblem efficiently. In particular, when $\mathcal{D}$ is an $\ell_1$ norm constraint as in problem \eqref{eq:sga}, a \textit{single} coordinate (corresponding to the largest entry in the magnitude of the current gradient) is greedily added to the solution found so far (Algorithm~\ref{alg:local}). Taking $\vct{\alpha}^{(0)} = \vct{0}$, we get $\|\tilde{\vct{\alpha}}\|_0 \leq 6.75C_f/\epsilon$. Combining this observation with Theorem~\ref{thm:fwconv} shows that Algorithm~\ref{alg:fw} always finds an $\epsilon$-approximate solution to problem \eqref{eq:sga} with $O(1/\epsilon)$ nonzero entries, which is sometimes called an $\epsilon$-coreset of size $O(1/\epsilon)$ \citep{Shalev-Shwartz2010,Jaggi2011}.
It turns out that this is worst-case optimal (up to constant terms) for the family of problems \eqref{eq:sga}, as evidenced by the derivation of a matching lower bound of $\Omega(1/\epsilon)$ for the sparsity of an $\epsilon$-approximate solution \citep{Jaggi2011}.\footnote{Note that the dependency on $\epsilon$ can be improved to $O(n\log(1/\epsilon))$ for strongly convex functions by using so-called away-steps (see for instance \citep{Lacoste-Julien2013a}). However, when $n$ is large, this improvement might not be desirable since it introduces a dependency on $n$.}

Another domain $\mathcal{D}$ of interest is the unit simplex $\Delta_n = \{\alpha\in\mathbb{R}^n : \vct{\alpha}\geq 0, \sum_i\alpha_i = 1\}$. The solution to the linear subproblem, shown in Algorithm~\ref{alg:local}, also ensures the sparsity of the iterates and leads to an $\epsilon$-approximation with $O(1/\epsilon)$ nonzero entries, which is again optimal \citep{Clarkson2010,Jaggi2013}.

For both domains, the number of nonzero entries depends only on $\epsilon$, and not on $n$, which is particularly useful for tackling large-scale problems. These results constitute the basis for our design of a distributed algorithm with small communication overhead. 

\section{Distributed Frank-Wolfe (dFW) Algorithm}
\label{sec:dfw}

\newcommand{\broadcast}{\textbf{broadcast} }
\newcommand{\term}{\textbf{terminate} }
\newcommand{\LINEIF}[2]{%
    \algorithmicif\ {#1}\ \algorithmicthen\ {#2} \algorithmicend\ \algorithmicif%
}
\begin{algorithm*}[t]
\caption{Distributed Frank-Wolfe (dFW) algorithm for problem \eqref{eq:sga}}
\label{alg:dfw}
\begin{algorithmic}[1]
\STATE on all nodes: $\vct{\alpha}^{(0)} = \vct{0}$
\FOR{$k = 0,1,2,\dots$}
\STATE on each node $v_i\in V$:
\begin{itemize}[noitemsep,topsep=2pt,itemindent=-10pt]
\item \hspace*{-0.5cm} $j^{(k)}_i = \argmax_{j\in \mathcal{A}_i} \left| \nabla f(\vct{\alpha}^{(k)})_j\right|$\hspace{0.5cm}{\it\small // find largest entry of the local gradient in absolute value}
\item \hspace*{-0.5cm} $S^{(k)}_i = \sum_{j\in\mathcal{A}_i} \alpha^{(k)}_j\nabla f(\vct{\alpha}^{(k)})_j$\hspace{2.5cm}{\it\small // compute partial sum for stopping criterion}
\item \hspace*{-0.5cm} \broadcast $g^{(k)}_i = \nabla f(\vct{\alpha}^{(k)})_{j^{(k)}_i}$ and $S^{(k)}_i$\hspace{6.48cm}
\end{itemize}
\STATE on each node $v_i\in V$:
\begin{itemize}[noitemsep,topsep=2pt,itemindent=-10pt]
\item \hspace*{-0.5cm} $i^{(k)}=\argmax_{i\in[N]} \left|g^{(k)}_i\right|$\hspace{2cm}{\it\small // compute index of node with largest overall gradient }
\item \hspace*{-0.5cm} \LINEIF{$i = i^{(k)}$}{\broadcast $j^{(k)} = j^{(k)}_i$ and $\vct{a}_{j^{(k)}}$}\hspace{0.2cm}{\it\small // send atom and index to other nodes}
\end{itemize}
\STATE on all nodes: $\vct{\alpha}^{(k+1)} = (1-\gamma) \vct{\alpha}^{(k)} + \gamma \sgn \left[-g^{(k)}_{i^{(k)}}\right]\beta\vct{e}^{j^{(k)}}$, where $\gamma=\frac{2}{k+2}$ or line-search
\ENDFOR
\STATE{{\bf stopping criterion}}: $\sum_{i=1}^N S^{(k)}_i+\beta\left|g^{(k)}_{i^{(k)}}\right| \leq \epsilon$\hspace{7.5cm}
\end{algorithmic}
\end{algorithm*}

We now focus on the problem of solving \eqref{eq:sga} in the distributed setting. We consider a set of $N$ local nodes $V = \{v_i\}_{i=1}^N$ which communicate according to an undirected connected graph $G = (V,E)$ where $E$ is a set of $M$ edges. An edge $(v_i,v_j)\in E$ indicates that $v_i$ and $v_j$ can communicate with each other. To simplify the analysis, we assume no latency in the network and synchronous updates --- our experiments in Section~\ref{sec:large} nonetheless suggest that these simplifying assumptions could be relaxed.

The atom matrix $\mat{A}\in\mathbb{R}^{d\times n}$ is partitioned column-wise across $V$. Formally, for $i\in [N]$, let the {\it local data} of $v_i$, denoted by $\mathcal{A}_i \subseteq [n]$, be the set of column indices associated with node $v_i$, where $\bigcup_i \mathcal{A}_i = [n]$ and $\mathcal{A}_i  \bigcap \mathcal{A}_j  = \emptyset$ for $i\neq j$. The {\it local gradient} of $v_i$ is the gradient of $f$ restricted to indices in $\mathcal{A}_i$. We measure the communication cost in number of real values transmitted.

\subsection{Basic Algorithm}

Our communication-efficient procedure to solve \eqref{eq:sga} in the distributed setting is shown in Algorithm~\ref{alg:dfw}.\footnote{Note that it can be straightforwardly adapted to deal with a simplex constraint based on Algorithm~\ref{alg:local}.} The algorithm is parameter-free and each iteration goes as follows: (i) each node first identifies the largest component of its local gradient (in absolute value) and broadcasts it to the other nodes (step 3), (ii) the node with the overall largest value broadcasts the corresponding atom, which will be needed by other nodes to compute their local gradients in subsequent iterations (step 4), and (iii) all nodes perform a FW update (step 5). We have the following result. 

\begin{thm}
\label{thm:dfw}
Algorithm~\ref{alg:dfw} terminates after $O(1/\epsilon)$ rounds and $O\left((Bd+NB)/\epsilon\right)$ total communication, where $B$ is an upper bound on the cost of broadcasting a real number to all nodes in the network $G$. At termination, all nodes hold the same $\tilde{\vct{\alpha}}$ with optimality gap at most $\epsilon$ and the $O(1/\epsilon)$ atoms corresponding to its nonzero entries.
\end{thm}
\begin{proof}
We prove this by showing that dFW performs Frank-Wolfe updates and that the nodes have enough information to execute the algorithm. The key observation is that the local gradient only depends on local data and atoms received so far. Details can be found in Appendix~\ref{app:dfw}.
\end{proof}

Theorem~\ref{thm:dfw} shows that dFW allows a trade-off between communication and optimization error without dependence on the total number of atoms, which makes the algorithm very appealing for large-scale problems. The communication cost depends only linearly on $1/\epsilon$, $d$ and the network topology.

\subsection{An Approximate Variant}
\label{sec:approx}

\begin{algorithm}[t]
\caption{GreedySelection($\mathcal{A}, \mathcal{C}, \Delta k$) \hfill // add $\Delta k$ centers from $\mathcal{A}$ to $\mathcal{C}$}
\label{alg:kcenter}
\begin{algorithmic}[1]
\STATE  \textbf{repeat} $\Delta k$ times
\STATE \quad $\mathcal{C} = \mathcal{C}\cup\{\vct{a}_{j'}\}$, where $j' = \argmax_{j\in \mathcal{A}} d(\vct{a}_j, \mathcal{C})$ with $d(\vct{a}_j, \mathcal{C}) = \min_{l \in \mathcal{C}} \|\vct{a}_j - \vct{a}_l \|_1$
\end{algorithmic}
\end{algorithm}

\begin{algorithm}[t]
\caption{Approximate dFW algorithm}
\label{alg:bdfw}
\begin{algorithmic}[1]
\STATE on each node $v_i\in V$: $\vct{\alpha}^{(0)} = \vct{0}$; 
$\mathcal{C}_i = \text{GreedySelection}(\emptyset, m_i^{(0)})$
\FOR{$k = 0,1,2,\dots$}
\STATE perform the same steps as in Algorithm~\ref{alg:dfw}, except that $j_i^{(k)}$ is selected from $\mathcal{C}_i$
\STATE on each node $v_i\in V$: $\mathcal{C}_i = \text{GreedySelection}(\mathcal{A}_i,\mathcal{C}_i, m_i^{(k+1)} - m_i^{(k)})$
\ENDFOR
\end{algorithmic}
\end{algorithm}

Besides its strong theoretical guarantees on optimization error and communication cost, dFW is very scalable, as the local computational cost of an iteration typically scales linearly in the number of atoms on the node. However, if these local costs are misaligned, the algorithm can suffer from significant wait time overhead due to the synchronized updates.
To address this issue, we propose an approximate variant of dFW which can be used to balance or reduce the effective number of atoms on each node. The algorithm is as follows. First, each node $v_i$ clusters its local dataset $\mathcal{A}_i$ into $m_i$ groups using the classic greedy $m$-center algorithm of \citet{gonzalez1985clustering}, which repeatedly adds a new center that is farthest from the current set of centers (Algorithm~\ref{alg:kcenter}). Then, dFW is simply run on the resulting centers. Intuitively, each original atom has a nearby center, thus selecting only from the set of centers is optimal up to small additive error that only leads to small error in the final solution. 
Optionally, we can also gradually add more centers so that the additive error scales as $O(1/k)$, in which case the error essentially does not affect the quality of the final solution.
The details are given in Algorithm~\ref{alg:bdfw}, and we have the following result.

\begin{lemma}
Let $r^{opt}(\mathcal{A}, m)$ denote the optimal $\ell_1$-radius of partitioning the local data indexed by $\mathcal{A}$ into $m$ clusters,
and let $r^{opt}(\vct{m}) := \max_i r^{opt}(\mathcal{A}_i, m_i)$. Let $G := \max_{\vct{\alpha}} \| \nabla g(\mat{A}\vct{\alpha}) \|_\infty$. Then, Algorithm~\ref{alg:bdfw} computes a feasible solution with optimality gap at most $\epsilon + O(G r^{opt}(\vct{m}^{0}) )$ after $O(1/\epsilon)$ iterations. Furthermore, if $r^{opt}(\vct{m}^{(k)}) = O(1/Gk)$, then the optimality gap is at most $\epsilon$.
\end{lemma}
\begin{proof}
By the analysis in~\cite{gonzalez1985clustering}, the maximum cluster radius on node $v_i$ is at most $2r^{opt}(\mathcal{A}_i, m^{(k)}_i)$ at time $k$, hence the maximum $\ell_1$-radius over all clusters is at most $2r^{opt}(\vct{m}^{(k)})$.
For any $j' \in [n]$, there exists $j \in \cup_i \mathcal{S}_i$ such that
$$|\nabla f(\vct{\alpha})_j - \nabla f(\vct{\alpha})_{j'}| = |\vct{a}_j \T\nabla g(\mat{A}\vct{\alpha}) - \vct{a}_{j'}\T\nabla g(\mat{A}\vct{\alpha})| \leq 2r^{opt}(\vct{m}^{(k)}) G.$$
Thus the selected atom is optimal up to a $2r^{opt} (\vct{m}^{(k)}) G$ additive error. Then the first claim follows from the analysis of \cite{Freund2013}.

Similarly, when $r^{opt}(\vct{m}^{(k)}) = O(1/Gk)$, the additive error is $O(1/k)$. The second claim then follows from the analysis of \cite{Jaggi2013}.
\end{proof}

Algorithm~\ref{alg:bdfw} has the potential to improve runtime over exact dFW in various practical situations. For instance, when dealing with heterogenous local nodes, $v_i$ can pick $m_i$ to be proportional to its computational power so as to reduce the overall waiting time. Another situation where this variant can be useful is when the distribution of atoms across nodes is highly unbalanced. We illustrate the practical benefits of this approximate variant on large-scale Kernel SVM in Section~\ref{sec:exp}.

\subsection{Illustrative Examples}
\label{sec:app}

We give three examples of interesting problems that dFW can solve efficiently: LASSO regression, Kernel SVM and $\ell_1$-Adaboost. Others applications include sparse coding \cite{Lee2006} and multiple kernel learning \cite{Bach2004}.

\paragraph{LASSO regression with distributed features}

LASSO \citep{Tibshirani1996} is a linear regression method which aims at approximating a target vector $\vct{y}\in\mathbb{R}^d$ by a sparse combination of features:
\begin{equation}
\label{eq:lasso}
\min_{\vct{\alpha}\in \mathbb{R}^n} \quad \|\vct{y} - \mat{A}\vct{\alpha}\|_2^2 \quad \text{s.t.} \quad \|\vct{\alpha}\|_1 \leq \beta,
\end{equation}
where $y_i$ and $a_{ij}$ are the target value and $j^{\text{th}}$ feature for training point $i$, respectively. Our dFW algorithm applies to problem \eqref{eq:lasso} in the context of distributed features (each local node holds a subset of the features for all training points), as common in distributed sensor networks \citep[see e.g.][]{Christoudias2008}. It applies in a similar manner to sparse logistic regression \citep{Shevade2003} and, perhaps more interestingly, to Group-LASSO \citep{Yuan2006}. For the latter, suppose each group of features is located on the same node, for instance when combining features coming from the same source in multi-view learning or when using dummy variables that encode categorical features. In this case, each atom is a group of features and dFW selects a single group at each iteration \citep[see][for the Frank-Wolfe update in this case]{Jaggi2013}.

%
%

\paragraph{Support Vector Machines with distributed examples}

Let $\{\vct{z}_i = (\vct{x}_i, y_i)\}_{i=1}^n$ be a training sample with $\vct{x}_i\in\mathbb{R}^d$ and $y_i \in \{\pm1\}$. Let $k(\vct{x},\vct{x}') = \innerp{\varphi(\vct{x}),\varphi(\vct{x}')}$ be a PSD kernel. The dual problem of L2-loss SVM with kernel $k$ is as follows \citep{Tsang2005,Ouyang2010}:
\begin{equation}
\label{eq:svmd}
\min_{\vct{\alpha}\in \Delta_n} \quad \vct{\alpha}\T\tilde{\mat{K}}\vct{\alpha},
\end{equation}
with $\tilde{\mat{K}} =  [\tilde{k}(\vct{z}_i,\vct{z}_j)]_{i,j=1}^n$ where $\tilde{k}(\vct{z}_i,\vct{z}_j) =y_iy_j k(\vct{x}_i,\vct{x}_j) + y_i y_j + \frac{\delta_{ij}}{C}$ is the ``augmented kernel''.

It is easy to see that $\tilde{k}(\vct{z}_i,\vct{z}_j) = \innerp{\tilde{\varphi}(\vct{z}_i),\tilde{\varphi}(\vct{z}_j)}$ with $\tilde{\varphi}(\vct{z}_i) = [y_i \varphi(\vct{x}_i), y_i, \frac{1}{\sqrt{C}} \vct{e}_i]$.
Thus $\tilde{\mat{K}} = \tilde{\mat{\Phi}}\T\tilde{\mat{\Phi}}$ where $\tilde{\mat{\Phi}} = [ \tilde{\varphi}(\vct{z}_1),\dots, \tilde{\varphi}(\vct{z}_n)]\T$ is the atom matrix. Notice however that atoms may be of very high (or even infinite) dimension since they lie in kernel space. Fortunately, the gradient only depends on $\tilde{\mat{K}}\in\mathbb{R}^{n\times n}$, thus we can broadcast the training point $\vct{z}_i$ instead of $\tilde{\varphi}(\vct{z}_i)$ (assuming all nodes know the kernel function $k$). Note that dFW can be seen as a distributed version of Core Vector Machines \citep{Tsang2005}, a large-scale SVM algorithm based on Frank-Wolfe, and that it applies in a similar manner to a direct multi-class L2-SVM \citep{Asharaf2007} and to Structural SVM \citep[see the recent work of][]{Lacoste-Julien2013}.

%
%

\paragraph{Boosting with distributed base classifiers}

Let $\mathcal{H}=\{h_j(\cdot) : \mathcal{X}\rightarrow\mathbb{R}, j=1,\dots,n\}$ be a class of base classifiers to be combined. Given a training sample $\{(\vct{x}_i, y_i)\}_{i=1}^d$, let $\mat{A}\in\mathbb{R}^{d\times n}$ where $a_{ij} = y_ih_j(\vct{x}_i)$. We consider the following $\ell_1$-Adaboost formulation \citep{Shen2010}:
\begin{equation}
\label{eq:boost}
\min_{\vct{\alpha}\in \Delta_n} \quad \log\left(\frac{1}{d}\sum_{i=1}^d\exp\left(-\frac{(\mat{A}\vct{\alpha})_i}{T}\right)\right),
\end{equation}
where $T>0$ is a tuning parameter. Our algorithm thus straightforwardly applies to \eqref{eq:boost} when base classifiers are distributed across nodes.
The gradient of \eqref{eq:boost} is as follows:
$$\nabla f(x) = -\mat{w}^T\mat{A},\quad\text{with } \vct{w} = \frac{\exp(-\mat{A}\vct{\alpha}/T)}{\sum_{i=1}^d\exp\left(-(\mat{A}\vct{\alpha})_i/T\right)}.$$
The Frank-Wolfe update thus corresponds to adding the base classifier that performs best on the training sample weighted by $\vct{w}$, where $\vct{w}$ defines a distribution that favors points that are currently misclassified.
This can be done efficiently even for a large (potentially infinite) family of classifiers when a weak learner for $\mathcal{H}$ is available.\footnote{A weak learner for $\mathcal{H}$ returns the base classifier in $\mathcal{H}$ that does best on the weighted sample defined by $\vct{w}$.} In particular, when $\mathcal{H}$ is the class of decision stumps and the features are distributed, each node can call the weak learner on its local features to determine its best local base classifier.


\newcommand{\OPT}{\mathrm{OPT}}
\newcommand{\card}[1]{\mathrm{card}({#1})}

\section{Theoretical Analysis}
\label{sec:theory}

\begin{figure}[t]
\centering
\includegraphics[width=0.6\columnwidth]{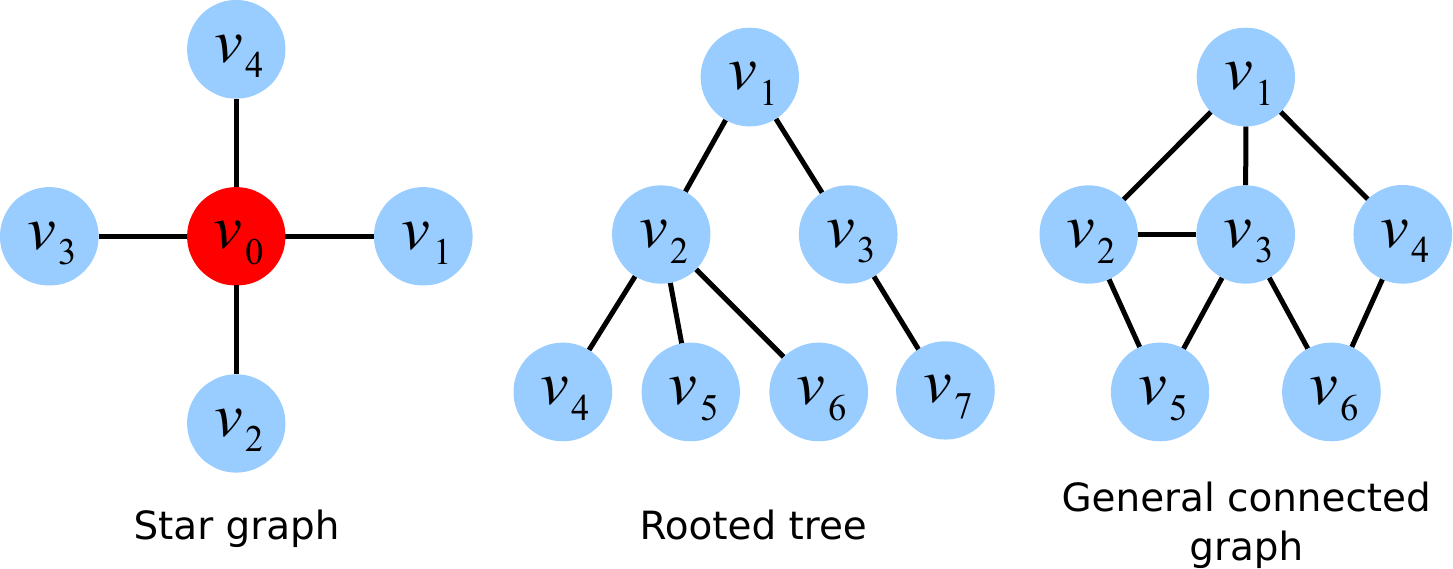}
\caption{Illustration of network topologies.}
\label{fig:network}
\vspace{-2mm}
\end{figure}

\subsection{Communication Cost of dFW under Different Network Topology}

We derive the communication cost of dFW for various types of network graphs, depicted in Figure~\ref{fig:network}. 


\paragraph{Star network}

In a star network, there is a coordinator node $v_0$ connected to all local nodes $v_1,\dots,v_N$.
Broadcasting a number to all nodes has cost $B=N$, so
the communication cost is
$O\left((Nd+N^2)/\epsilon\right)$.
This can be improved to $O(Nd/\epsilon)$: instead of broadcasting $g^{(k)}_i$ and $S^{(k)}_i$, we can send these quantities only to $v_0$
which selects the maximum $g^{(k)}_i$ and computes the sum of $S^{(k)}_i$.

\paragraph{Rooted tree}
Broadcasting on a tree costs
$B=O(N)$, but this can be improved again by avoiding broadcast of $g^{(k)}_i$ and
$S^{(k)}_i$. To select the maximum $g^{(k)}_i$, each node sends to its parent the maximum value among its local
$g^{(k)}_i$ and the values received from its children. The
root gets the maximum $g^{(k)}_i$ and can send it back to all
the nodes. A similar trick works for computing the sum of
$S^{(k)}_i$. Hence, only $j^{(k)}$ and the selected atom need to be broadcasted, and the total
communication is $O(Nd/\epsilon)$.

\paragraph{General connected graph}
If the nodes can agree on a spanning tree, then this reduces to
the previous case.
We further consider the case when nodes operate without using any information beyond that in their local neighborhood and cannot agree on a spanning tree \citep[this is known as the {\it fully distributed
setting},][]{mosk2010fully}.
In this case, broadcasting can be done through a message-passing procedure similar to the one used in \citet{Balcan2013}.
The cost of broadcasting a number is $B=O(M)$, and thus the
total communication cost is $O\left(M(N+d)/\epsilon\right)$.


\subsection{Lower Bound on Communication Cost}

In this section, we derive a lower bound on the communication required by any deterministic algorithm to construct an $\epsilon$-approximation to problem \eqref{eq:sga} for either the $\ell_1$-ball or the simplex constraint. We assume that at the end of the algorithm, at least one of the nodes has all the selected atoms. This is not a restrictive assumption since in many applications, knowledge of the selected atoms is required in order to use the learned model (for instance in kernel SVM).

Our proof is based on designing an optimization problem that meets the desired minimum amount of communication.
First, we consider a problem for which any $\epsilon$-approximation solution must have $O(1/\epsilon)$ selected atoms. We split the data across two nodes $v_1$ and $v_2$ so that these atoms must be almost evenly split across the two nodes.
Then, we show that for any fix dataset on one node, there are $T$ different instances of the dataset on the other node, so that in any two such instances, the sets of selected atoms are different.
Therefore, for any node to figure out the atom set, it needs $O(\log T)$ bits. The proof is completed by showing $\log T = \Omega(d/\epsilon)$.

\subsubsection{Setup}

Consider the problem of minimizing $f(\vct{\alpha}):=\|\mat{A} \vct{\alpha} \|_2^2$ over the unit simplex $\Delta_d$, where $\mat{A}$ is a $d\times d$ orthonormal matrix (i.e., $\mat{A}^T \mat{A} = \mat{I}$). Suppose the first $d/2$ columns of $\mat{A}$ (denoted as $\mat{A}_1$) are on the node $v_1$ and the other half (denoted as $\mat{A}_2$) are on the node $v_2$. Restrict $\mat{A}$ to be block diagonal as follows: vectors in $\mat{A}_1$ only have non-zeros entries in coordinates $\{1,\dots, d/2\}$ and vectors in $\mat{A}_2$ only have non-zeros entries in coordinates $\{d/2+1,\dots, d\}$.

For any deterministic algorithm $\mathfrak{A}$, let $\mathfrak{A}(\mat{A}_1,\mat{A}_2)$ denote the set of atoms selected given input $\mat{A}_1$ and $\mat{A}_2$, and let $f_\mathfrak{A}(\mat{A}_1,\mat{A}_2)$ denote the minimum value of $f(\vct{\alpha})$ when $\vct{\alpha} \in \Delta_d$ and $\vct{\alpha}$ has only non-zero entries on $\mathfrak{A}(\mat{A}_1,\mat{A}_2)$. We only consider deterministic algorithms that lead to provably good approximations: $\mathfrak{A}$ satisfies $f_\mathfrak{A}(\mat{A}_1,\mat{A}_2) < \epsilon + \OPT$ for any $\mat{A}_1,\mat{A}_2$. Furthermore, at the end of the algorithm, at least one of the nodes has all the selected atoms.

In practice, bounded precision representations are used for real values.
Study of communication in such case requires algebraic assumptions~\citep{abelson1980lower,tsitsiklis1987communication}.
For our problem, we allow the columns in $\mat{A}$ to be orthogonal and of unit length only approximately. More precisely, we say that two vectors $\vct{v}_1$ and $\vct{v}_2$ are near-orthogonal if $|\vct{v}_1^\top \vct{v}_2|\leq \frac{\epsilon}{4d^2}$, and a vector $\vct{v}$ has near-unit length if $|\vct{v}^\top \vct{v} -1| \leq \frac{\epsilon}{4d^2}$. The notion of subspace is then defined based on near-orthogonality and near-unit length. We make the following assumption about the representation precision of vectors in subspaces.

\begin{assumption}\label{ass:unit}
There exists a constant $\kappa > 1$ s.t. for any $t$-dimensional subspace, $t \leq d$,
the number of different $d$-dimensional near-unit vectors in the subspace is $\Theta(\kappa^{t-1})$.
\end{assumption}

\subsubsection{Analysis}

Let $\epsilon \in (0,1/6)$ and $d=\frac{1}{6\epsilon}$ (for simplicity, suppose $d/4$ is an integer).
We first show that the atoms in any good approximation solution must be almost evenly split across the two nodes.

\begin{claim}\label{cla:nnz}
Any $\vct{\hat\alpha}$ with $f(\vct{\hat\alpha}) < \epsilon + \OPT$ has more than $3d/4$ non-zero entries. Consequently, $|\mathfrak{A}(\mat{A}_p, \mat{A}_q)\cap \mat{A}_p| > d/4$ for $\{p,q\} = \{1,2\}$.
\end{claim}

\begin{proof}
First note that $f(\vct{\alpha}) = \vct{\alpha}^T \mat{A}^T \mat{A} \vct{\alpha}$. When using exact real values, $\mat{A}^T \mat{A}=\mat{I}$, thus the problem is equivalent to minimizing $g(\vct{\alpha}):= \|\vct{\alpha}\|_2^2$ over the unit simplex. When using bounded precision representation, the exact equivalence does not hold. Still, by our definitions of near-unit vectors and near-orthogonal vectors,
$
	\mat{A}^T \mat{A}=\mat{I} + \mat{E},
$
where each entry in $\mat{E}$ has absolute value at most $\frac{\epsilon}{4d^2}$. So for any $\vct{\alpha}$, $|f(\vct{\alpha}) - g(\vct{\alpha})| = |\vct{\alpha}^T \mat{E} \vct{\alpha}| \leq \frac{\epsilon}{4}$.

Then, any $\vct{\hat\alpha} \in \Delta_d$ with $f(\vct{\hat\alpha}) < \epsilon + \OPT$ satisfies that $g(\vct{\hat\alpha}) < 2\epsilon + \OPT_g$ where $\OPT_g$ is the optimal value for $g(\vct{\alpha})$. To see this, let $\vct{\alpha}_f$ denote the optimal solution for $f$ and $\vct{\alpha}_g$ denote the optimal solution for $g$. We have:
$$
g(\vct{\hat\alpha}) - \frac{\epsilon}{4} \leq f(\vct{\hat\alpha})  \leq \epsilon + f(\vct{\alpha}_f) \leq \epsilon + f(\vct{\alpha}_g) \leq \epsilon + g(\vct{\alpha}_g) + \frac{\epsilon}{4},
$$
which leads to $g(\vct{\hat\alpha}) < 2\epsilon + \OPT_g$.
By Lemma 3 of \citet{Jaggi2013}, we have that for $k \leq d$,
$$
	\min_{\vct{\alpha} \in \Delta_d, \card{\vct{\alpha} }  \leq k} g(\vct{\alpha}) = 1/k.
$$
So $\OPT_g=1/d$. Hence $ 1/\card{\vct{\hat\alpha}} < 2\epsilon + 1/d$, which leads to the claim.
\end{proof}

Next, we show that for any fix dataset on one node, there are many different instances of the dataset on the other node, so that any two different instances lead to different selected atoms.

\begin{claim}\label{cla:more}
For any instance of $\mat{A}_1$, there exist
$T = {\kappa^{\Omega(d/\epsilon)} } / { (d/e)^{O(d)} }$
different instances of $\mat{A}_2$ (denoted as $\{\mat{A}^i_2\}_{i=1}^T$)
such that the set of atoms selected from $\mat{A}^i_2$
are different for all $i$, i.e., for any $0\leq i\neq j \leq T$,
$$
	\mathfrak{A}(\mat{A}_1, \mat{A}^i_2) \cap \mat{A}^i_2 \neq \mathfrak{A}(\mat{A}_1, \mat{A}^j_2) \cap \mat{A}^j_2.
$$
\end{claim}
\begin{proof}
Let $\mathcal{N}_D(t)$ denote the number of different sets of $t$ near-unit vectors in $D$ dimension
that are near-orthogonal to each other. By Claim~\ref{cla:nnz}, if $\mathfrak{A}(\mat{A}_1, \mat{A}^i_2) \cap \mat{A}^i_2 = \mathfrak{A}(\mat{A}_1, \mat{A}^j_2) \cap \mat{A}^j_2$,
then $|\mat{A}^i_2 \setminus \mat{A}^j_2| \leq d/4$.
Then there are at most $\mathcal{N}_{d/2}(d/4)$ different $\mat{A}^i_2$ so that $\mathfrak{A}(\mat{A}_1, \mat{A}^i_2) \cap \mat{A}^i_2$ are all the same.
Since there are $\mathcal{N}_{d/2}(d/2)$ different instances of $\mat{A}_2$, we have 
$$
	T \geq \frac{\mathcal{N}_{d/2}(d/2)}{ \mathcal{N}_{d/2}(d/4)}.
$$

Next, we observe that $\mathcal{N}_{D}(t)=\Theta(\kappa^{(2D-t-1)t/2}/ t!)$. We first choose a near-unit vector in the $D$-dimensional subspace, and then choose a near-unit vector in the $(D-1)$-dimensional subspace orthogonal to the first chosen vector, and continue in this way until $t$ near-unit vectors are selected. By Assumption~1, there are $\Theta(\kappa^{(2D-t-1)t/2})$ different choices, and there are $t!$ repetitions, which leads to the bound on $\mathcal{N}_{D}(t)$. The bound on $T$ then follows from $d=\frac{1}{6\epsilon}$. 
\end{proof}

Distinguishing between the $T$ instances requires communication, resulting in our lower bound.

\begin{thm}\label{thm:lower}
Suppose a deterministic algorithm $\mathfrak{A}$ satisfies $f_\mathfrak{A}(\mat{A}_1,\mat{A}_2) < \epsilon + \OPT$ for any $\mat{A}_1,\mat{A}_2$.
Under Assumption~\ref{ass:unit}, the communication cost of $\mathfrak{A}$ is lower-bounded by $\Omega(\frac{d}{\epsilon} \log \kappa - d\log d)$ bits.
\end{thm}
\begin{proof}
Suppose the algorithm uses less than $\log T$ bits. At the end of the algorithm, if node $v_1$ has the selected atoms, then $v_1$ must have determined these atoms based on $\mat{A}_1$ and less than $\log T$ bits. By Claim~\ref{cla:more}, there are $T$ different instances of $\mat{A}_2$ with different $\mathfrak{A}(\mat{A}_1, \mat{A}_2) \cap \mat{A}_2$. For at least two of these instances,  $v_1$ will output the same set of atoms, which is contradictory. A similar argument holds if node $v_2$ has the selected atoms. Therefore, the algorithm uses at least $\log T$ bits.
\end{proof}

\subsubsection{Remarks} The same lower bound holds for $\ell_1$-constrained problems by reduction to the simplex case. The analysis also extends to networks with more than two nodes, in which case the diameter of the network is a multiplicative factor in the bound. Details can be found in Appendix~\ref{app:bound}.

The upper bound given by dFW in Theorem~\ref{thm:dfw} matches the lower bound of Theorem~\ref{thm:lower} in its dependency on $\epsilon$ and $d$. This leads to the important result that the communication cost of dFW is worst-case optimal in these two quantities. 

\section{Related Work}
\label{sec:related}


\subsection{Distributed Examples}

An important body of work considers the distributed minimization of a function $f(\vct{\alpha}) = \sum_{i=1}^N f_i(\vct{\alpha})$ where $f_i$ is convex and only known to node $i$. This typically corresponds to the risk minimization setting with distributed training examples, where $f_i$ gives the loss for the data points located on node $i$. These methods are based on subgradient descent (SGD) \citep[][and references therein]{Duchi2012,Dekel2012}, the Alternating Direction Methods of Multipliers (ADMM) \citep{Boyd2011,Wei2012} and Stochastic Dual Coordinate Ascent (SDCA) \citep{Yang2013}. Their rate of convergence (and thus their communication cost) for general convex functions is in $O(1/\epsilon)$ or $O(1/\epsilon^2)$, sometimes with a dependency on $n$, the total number of data points. 

Note that for Kernel SVM (Section~\ref{sec:app}), none of the above methods can be applied efficiently, either because they rely on the primal (GD, ADMM) or because they require broadcasting at each iteration a number of data points that is at least the number of nodes in the network (SDCA). In contrast, dFW only needs to broadcast a single point per iteration and still converges in $O(1/\epsilon)$, with no dependency on $n$. We illustrate its performance on this task in the experiments.

\subsection{Distributed Features}

In the case of distributed features, ADMM is often the method of choice due to its wide applicability and good practical performance \citep{Boyd2011}.\footnote{Note that to the best of our knowledge, no convergence rate is known for ADMM with distributed features.}
Given $N$ nodes, the parameter vector $\vct{\alpha}$ is partitioned as $\vct{\alpha} = [\vct{\alpha}_1,\dots,\vct{\alpha}_N]$ with $\vct{\alpha}_i \in \mathbb{R}^{n_i}$, where $\sum_{i=1}^N n_i = n$. The matrix $\mat{A}$ is partitioned conformably as $\mat{A} = [\mat{A}_1 \dots \mat{A}_N]$ and the optimization problem has the form
\begin{equation}
\label{eq:admm}
\begin{aligned}
\min_{\vct{\alpha}\in \mathbb{R}^n} &&& g\left(\sum_{i=1}^N\mat{A}_i\vct{\alpha}_i - \vct{y}\right) + \lambda R(\vct{\alpha}),
\end{aligned}
\end{equation}
where the loss $g$ is convex, the regularizer $R$ is separable across nodes and convex, and $\lambda>0$. Each node $v_i$ iteratively solves a local subproblem and sends its local prediction $\mat{A}_i\vct{\alpha}_i$ to a coordinator node, which then broadcasts the current global prediction $\sum_{i=1}^N\mat{A}_i\vct{\alpha}_i$. ADMM can be slow to converge to high accuracy, but typically converges to modest accuracy within a few tens of iterations \citep{Boyd2011}.

Both ADMM and dFW can deal with LASSO regression with distributed features. When features are dense, dFW and ADMM have a similar iteration cost but dFW is typically slower to converge as it only adds one feature to the solution at each iteration. On the other hand, the case of sparse features creates an interesting tradeoff between the number of iterations and the communication cost per iteration. Indeed, an iteration of dFW becomes much cheaper in communication than an iteration of ADMM since local/global predictions are typically dense. This tradeoff is studied in the experiments.

\subsection{Relation to Matching Pursuit}

Finally, note that for the $\ell_1$-constrained case, matching pursuit \citep{Tropp2004} greedily selects atoms in the same way as FW, although the updates are different \citep{Jaggi2011}. The proposed communication schemes thus readily apply to a distributed matching pursuit algorithm.

\section{Experiments}
\label{sec:exp}

We have shown that dFW enjoys strong theoretical guarantees on the communication cost. 
In this section, we validate our analysis with a series of experiments on two tasks: LASSO regression with distributed features, and Kernel SVM with distributed examples. In Section~\ref{sec:base}, we compare dFW to two baseline strategies. Section~\ref{sec:admm} studies the communication tradeoff between dFW and ADMM, a popular distributed optimization technique. Finally, Section~\ref{sec:large} proposes an evaluation of the performance of (exact and approximate) dFW in a real-world distributed environment.

%
%

\subsection{Comparison to Baselines}
\label{sec:base}

We first compare dFW to some baseline strategies for selecting a subset of atoms locally. We investigate two methods: (i) the random strategy, where each node selects a fixed-size subset of its atoms uniformly at random, and (ii) the local FW strategy, where each node runs the FW algorithm locally to select a fixed-sized subset as proposed by \citet{Lodi2010}. The objective function on the union of the selected subsets is then optimized with a batch solver. 
We compare the methods based on the objective value they achieve for a given communication cost.
In these experiments, we use a star graph with $N=100$ nodes and each atom is assigned to a node uniformly at random. Results are averaged over 5 runs.

\begin{figure*}[t]
\centering
\subfigure[Kernel SVM, Adult dataset]{
\includegraphics[width=0.35\columnwidth]{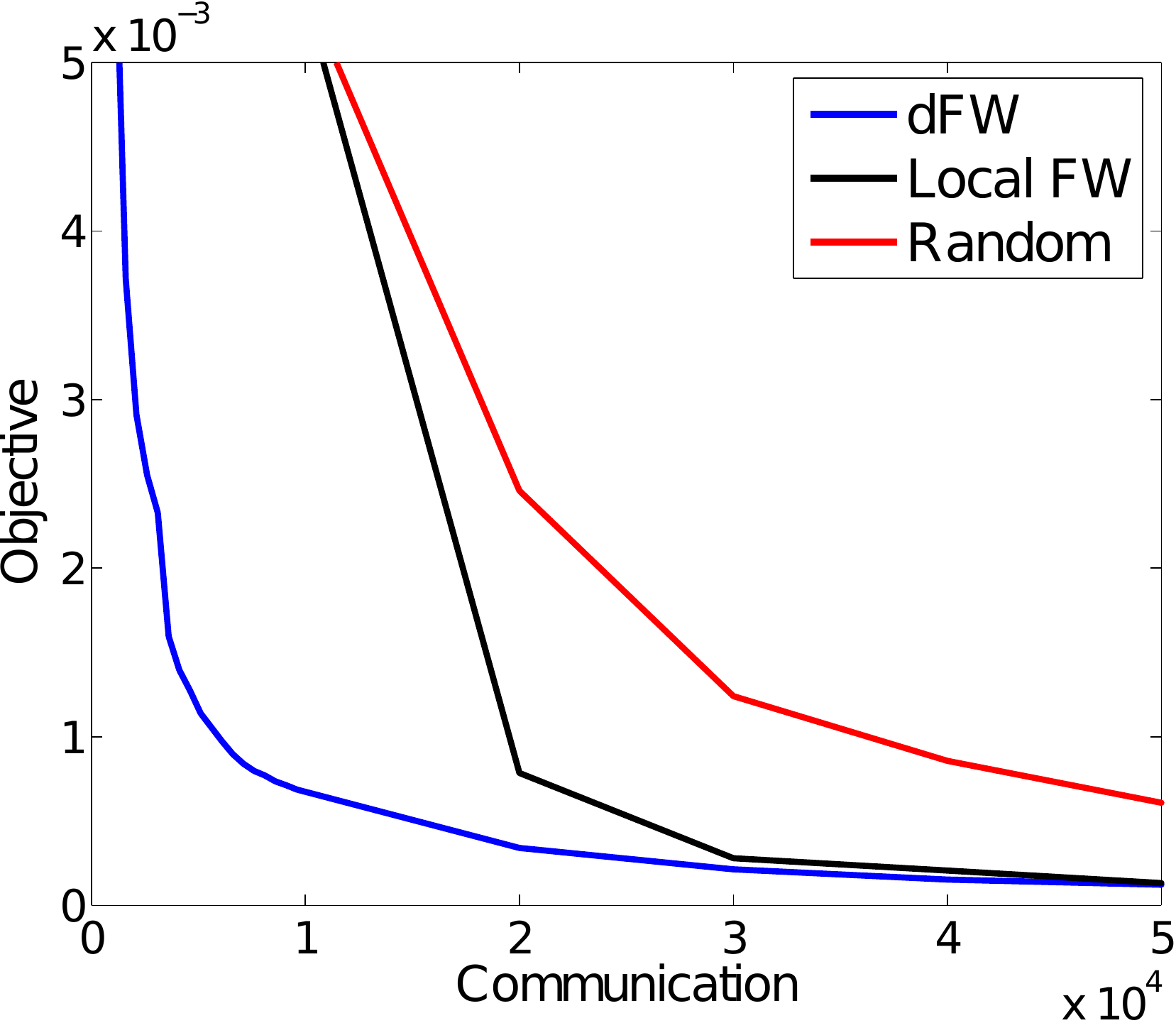}
\label{fig:svmadult}
}
\subfigure[LASSO, Dorothea dataset]{
\includegraphics[width=0.35\columnwidth]{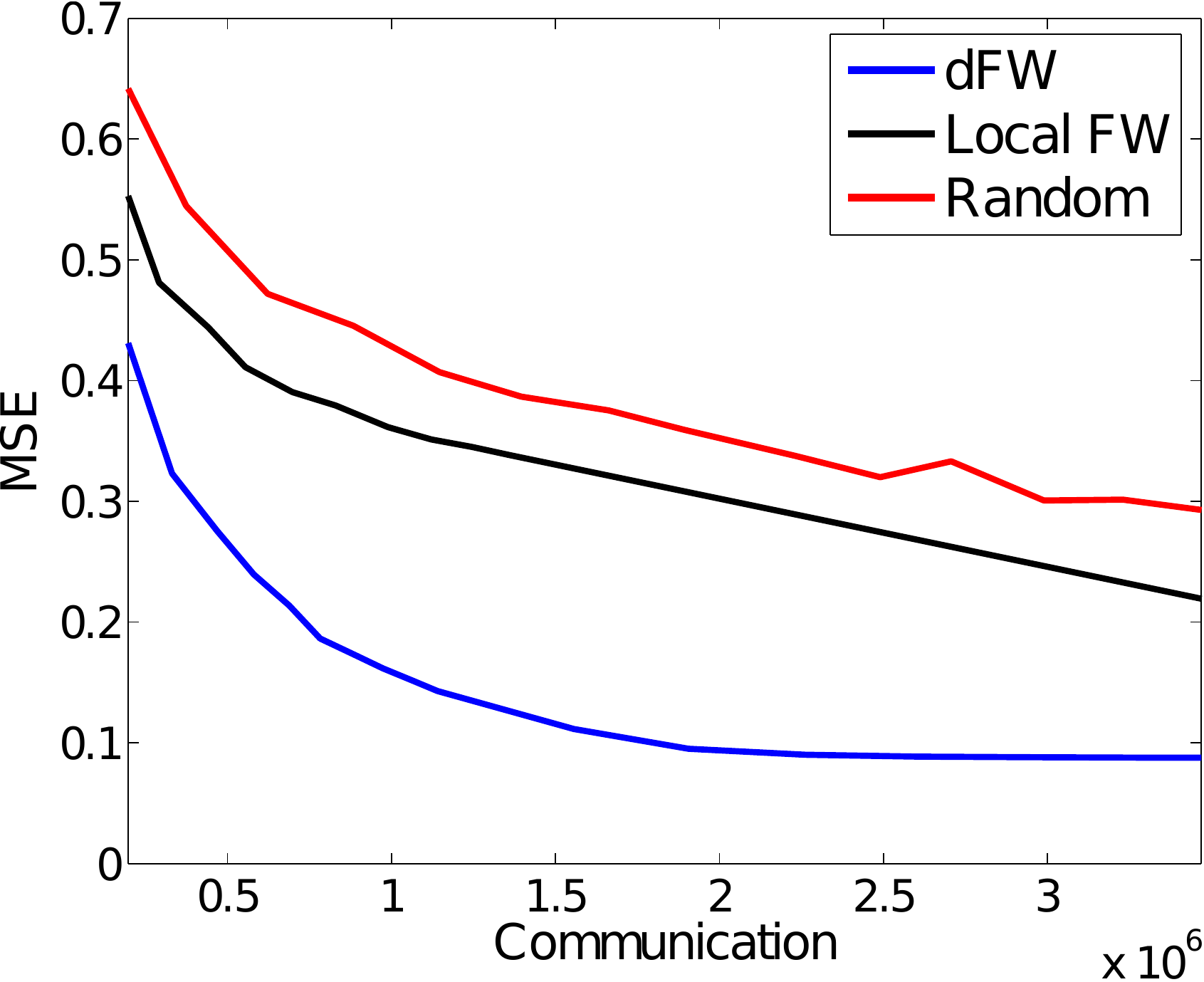}
\label{fig:lassobase}
}
\caption{Comparison with baseline strategies (best seen in color).}
\vspace{-2mm}
\end{figure*}

\paragraph{Kernel SVM with distributed examples}

We conduct experiments on the UCI Adult dataset, a census income database with about 32,000 examples and 123 features, so each node holds about 320 examples. We use the RBF kernel and set its bandwidth based on the averaged distance among examples and fix SVM parameter $C$  to $100$. Figure~\ref{fig:svmadult} shows that the training examples selected by dFW lead to significantly smaller objective value than those selected by both baselines. Note that the local FW does much better than the random strategy but performs worse than dFW due to its local selection criterion.

\paragraph{LASSO regression with distributed features}

We also conduct experiments on Dorothea, a drug discovery dataset from the NIPS 2003 Feature Selection Challenge \citep{Guyon2004}. It has 1,150 examples and 100,000 binary features (half of which are irrelevant to the problem), so each node holds about 1,000 features.
Figure~\ref{fig:lassobase} shows the results for $\beta=16$ (other values of $\beta$ do not change the relative performance of the methods significantly). Again, dFW significantly outperforms both baselines. On this dataset, these local strategies decrease the objective at a slow rate due the selection of irrelevant features.

\subsection{Comparison to ADMM}
\label{sec:admm}

In this section, we compare dFW with distributed ADMM \citep{Boyd2011}, a popular competing method, on the problem of LASSO with distributed features, both on synthetic and real data. As in the previous experiment, we use $N=100$ nodes and a uniform distribution of atoms, with results averaged over 5 runs.

\paragraph{Synthetic Data}

\begin{figure}[t]
\centering
\includegraphics[width=0.8\columnwidth]{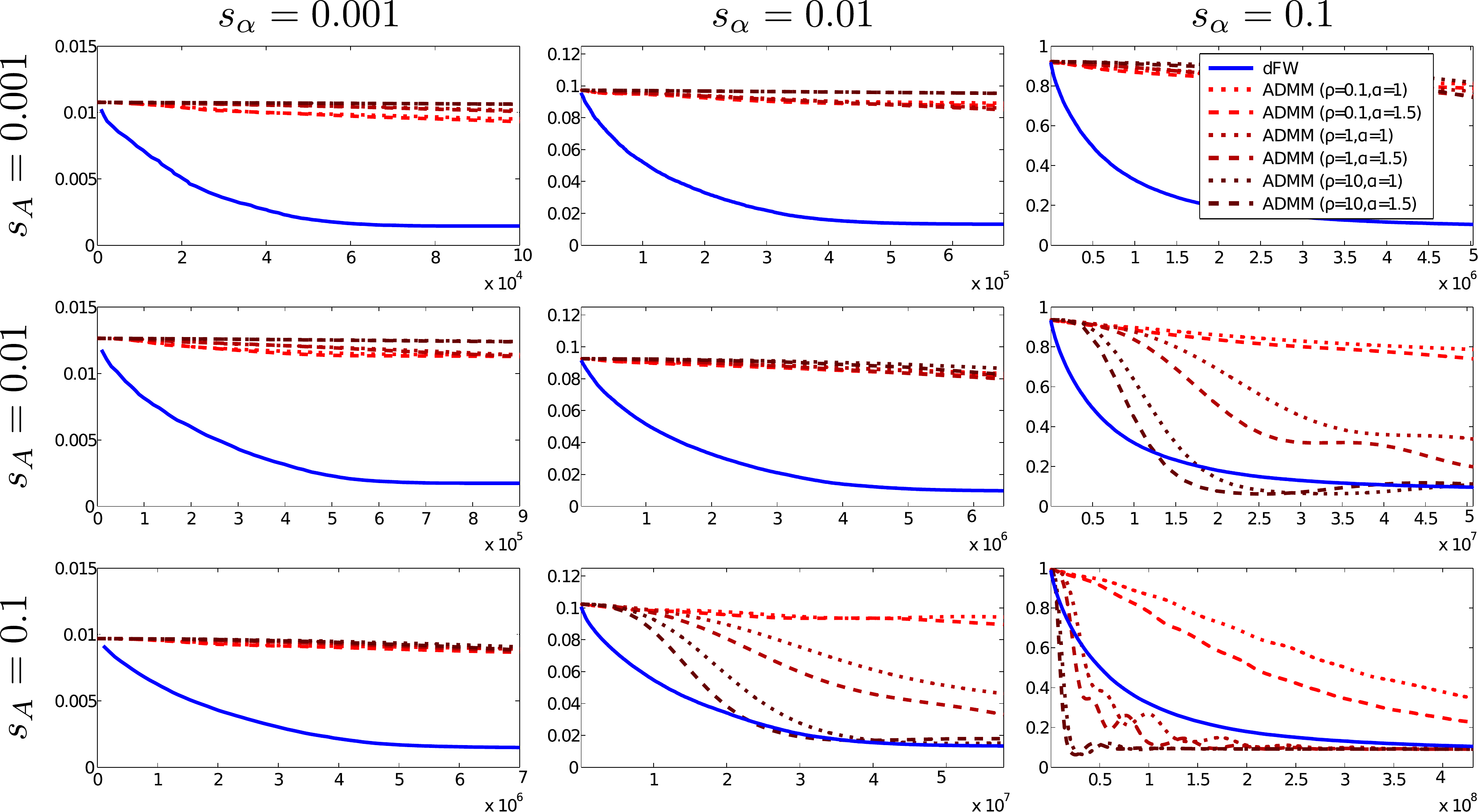}
\caption{Comparison with ADMM on LASSO using synthetic data of varying sparsity levels.
The plots show MSE versus communication.}
\label{fig:synth}
\end{figure}

The purpose of this experiment is to study the tradeoff between dFW and ADMM discussed in Section~\ref{sec:related}. To do so, we generate synthetic data following the protocol in \citep{Boyd2011} but with varying density of the data and the solution. 
Specifically, we create a data matrix $\mat{A}$ has $d=10,000$ examples and $n=100,000$ features (atoms) with density $s_A$ and its nonzero entries are drawn from $\mathcal{N}(0,1)$. The true solution $\vct{\alpha}^{true}$ has density $s_\alpha$ with each nonzero entry drawn from $\mathcal{N}(0,1)$, and the target vector is computed as $\vct{y} = \mat{A}\vct{\alpha}^{true} + \eta$, where $\eta\sim\mathcal{N}(0,10^{-3}\mat{I})$. We use $s_A,s_\alpha\in\{0.001,0.01,0.1\}$, thereby creating a set of 9 problems.
We set the regularization parameter to $\lambda=0.1\lambda_{max}$, where $\lambda_{max} = \|\mat{A}\vct{y}\|_\infty$ is the value above which the solution becomes the zero vector.\footnote{This choice results in a solution with a number of nonzero entries in the same order of magnitude as $s_\alpha$.} The corresponding value of $\beta$ for the constrained problem is obtained from the norm of the solution to the $\lambda$-regularized problem. For ADMM, we try several values of its penalty and relaxation parameters $\rho_{ADMM}$ and $\alpha_{ADMM}$, as done in \citep{Boyd2011}.

Figure~\ref{fig:synth} shows the MSE versus communication. These results confirm that dFW requires much less communication when the data and/or the solution are sparse (with no parameter to tune), while ADMM has the edge in the dense case (if its parameters are properly tuned). The rule of thumb seems to be that they perform similarly when $s_As_\alpha n = O(100)$, which is consistent with their communication cost per iteration and practical convergence speed.

\paragraph{Real data}

\begin{figure*}[t]
\centering
\includegraphics[width=0.32\columnwidth]{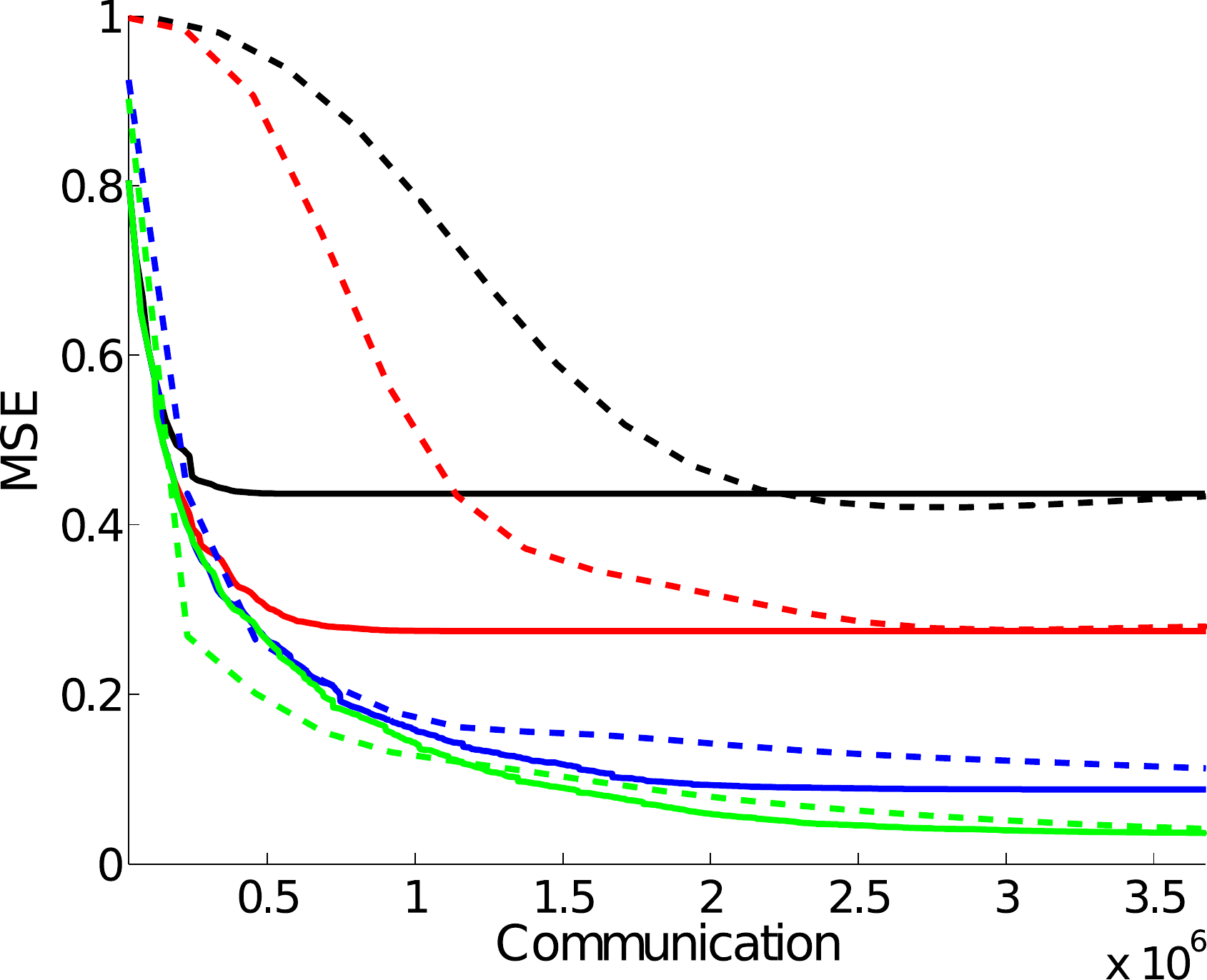}
\includegraphics[width=0.32\columnwidth]{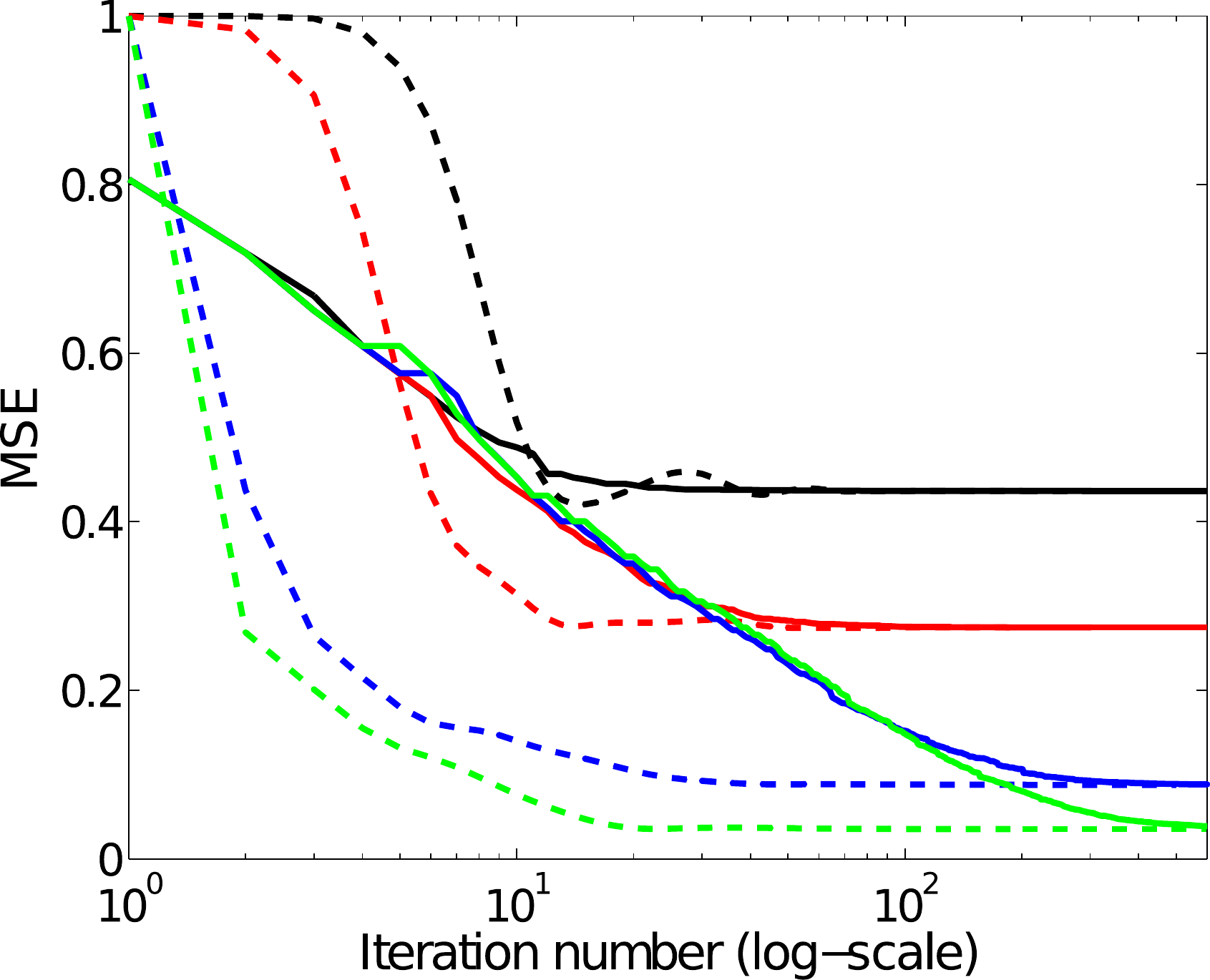}
\includegraphics[width=0.32\columnwidth]{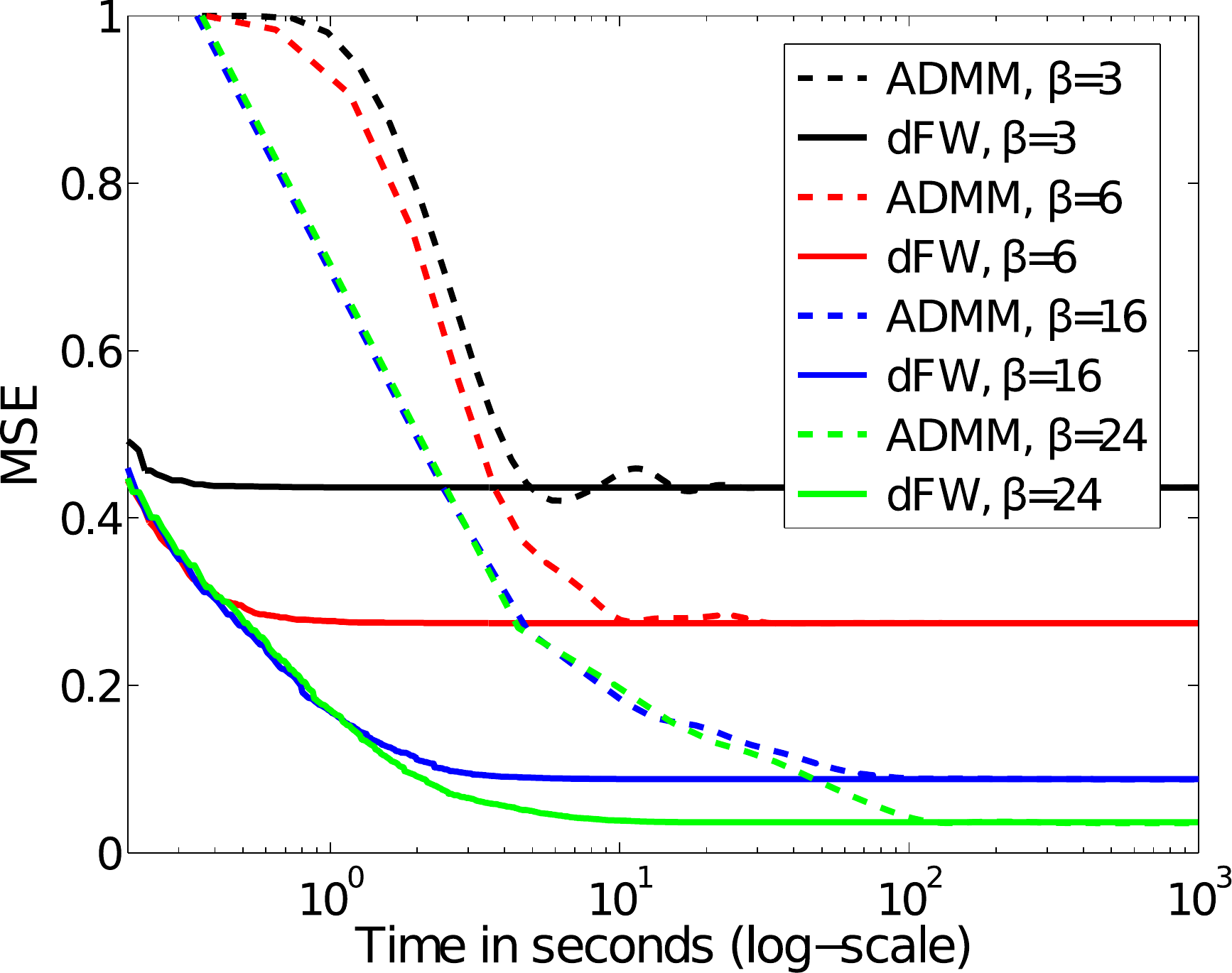}
\caption{Comparison with ADMM on LASSO, Dorothea dataset (best seen in color).}
\label{fig:lassoadmm}
\vspace{-2mm}
\end{figure*}

We also conduct experiments on the Dorothea dataset used in Section~\ref{sec:base}.
We compare dFW to ADMM for different values of $\beta$, which affect the sparsity of the solution. The results shown for ADMM are the best over its penalty $\rho\in \{0.1,1,10\}$ and relaxation $\alpha\in\{1,1.5\}$ parameters. 

The results in terms of communication cost shown in Figure~\ref{fig:lassoadmm} (left plot) are consistent with those obtained on synthetic data. Indeed, dFW is better at finding sparse solutions with very little communication overhead, even though many features in this dataset are irrelevant. On the other hand, as the target model becomes more dense ($\beta=24$), ADMM catches up with dFW. 

Figure~\ref{fig:lassoadmm} (middle plot) shows that ADMM generally converges in fewer iterations than dFW. However, an ADMM iteration is computationally more expensive as it requires to solve a LASSO problem on each node. In contrast, a dFW iteration amounts to computing the local gradients and is typically much faster. To give a rough idea of the difference, we show the time spent on sequentially solving the local subproblems on a single 2.5GHz CPU core in Figure~\ref{fig:lassoadmm} (right plot).\footnote{Both algorithms are implemented in Matlab. For ADMM, LASSO subproblems are solved using a proximal gradient method.} This shows that on this dataset, dFW decreases the objective faster than ADMM by 1 to 2 orders of magnitude despite its larger number of iterations.

\subsection{Large-scale Distributed Experiments}
\label{sec:large}

In this section, we evaluate the performance of dFW on large-scale kernel SVM in real-world distributed conditions.
Note that dFW is computationally very efficient if implemented carefully. In particular, the local gradient of a node $v_i$ can be recursively updated using only the kernel values between its $n_i$ points and the point received at the previous iteration. Thus the algorithm only needs $O(n_i)$ memory and the computational cost of an iteration is $O(n_i)$. We implemented dFW in C++ with openMPI and experiment with a fully connected network with $N\in\{1, 5, 10, 25, 50\}$ nodes which communicate over a 56.6-gigabit infrastructure. Each node is a single 2.4GHz CPU core of a separate host. We use a speech dataset of 8.7M examples with 120 features and 42 class labels (the task is majority class versus rest).\footnote{Data is from IARPA Babel Program, Cantonese language collection release IARPA-babel101b-v0.4c limited data pack.} In all experiments, we use the RBF kernel and set its bandwidth based on the averaged distance among examples. The parameter $C$ of SVM is set to $100$.

\begin{figure*}[t]
\centering
\subfigure[Exact dFW on ideal distribution]{
\includegraphics[width=0.33\columnwidth]{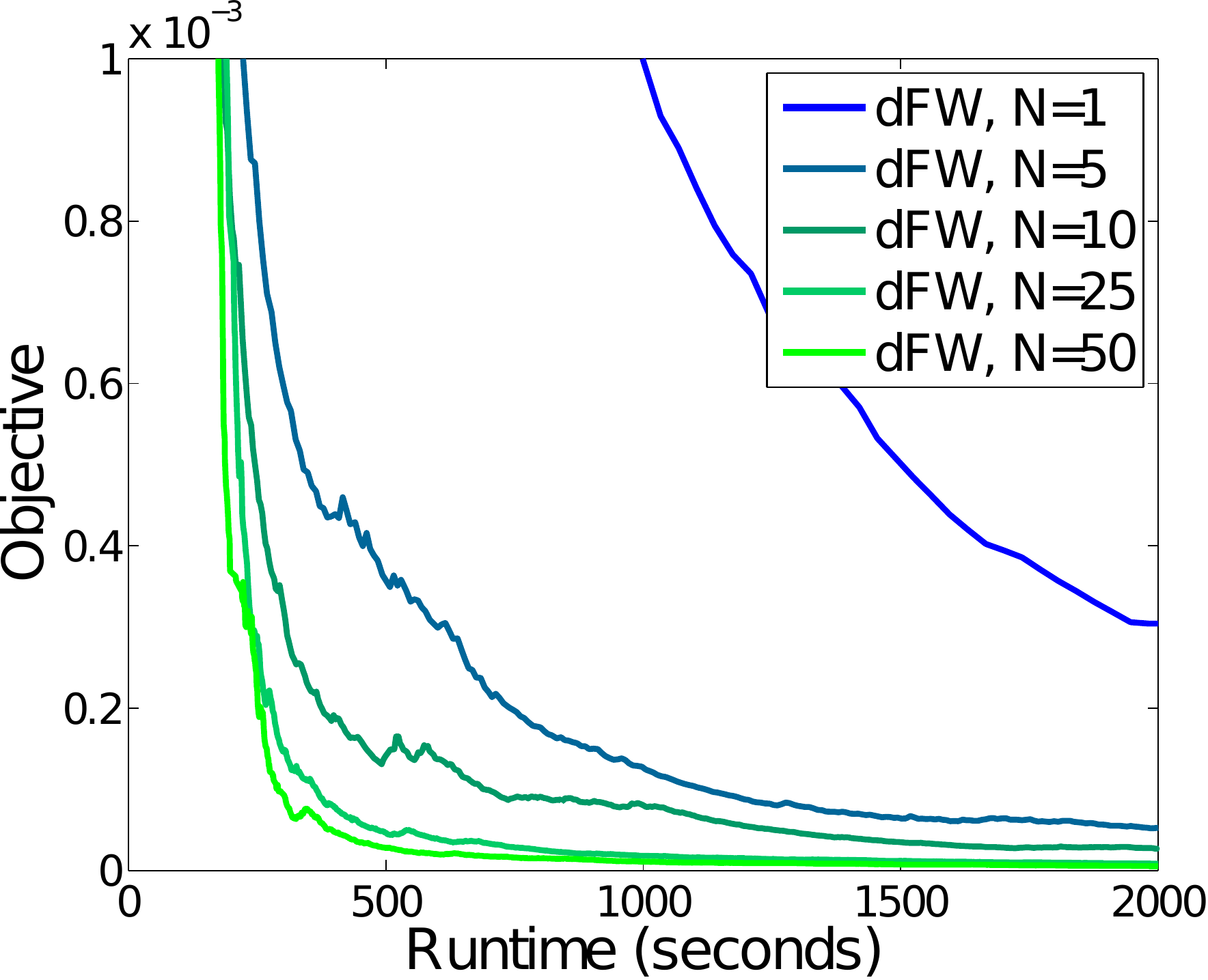}\hspace{-0.3cm}
\label{fig:mpi}
}
\subfigure[Approx. dFW to balance costs]{
\includegraphics[width=0.33\columnwidth]{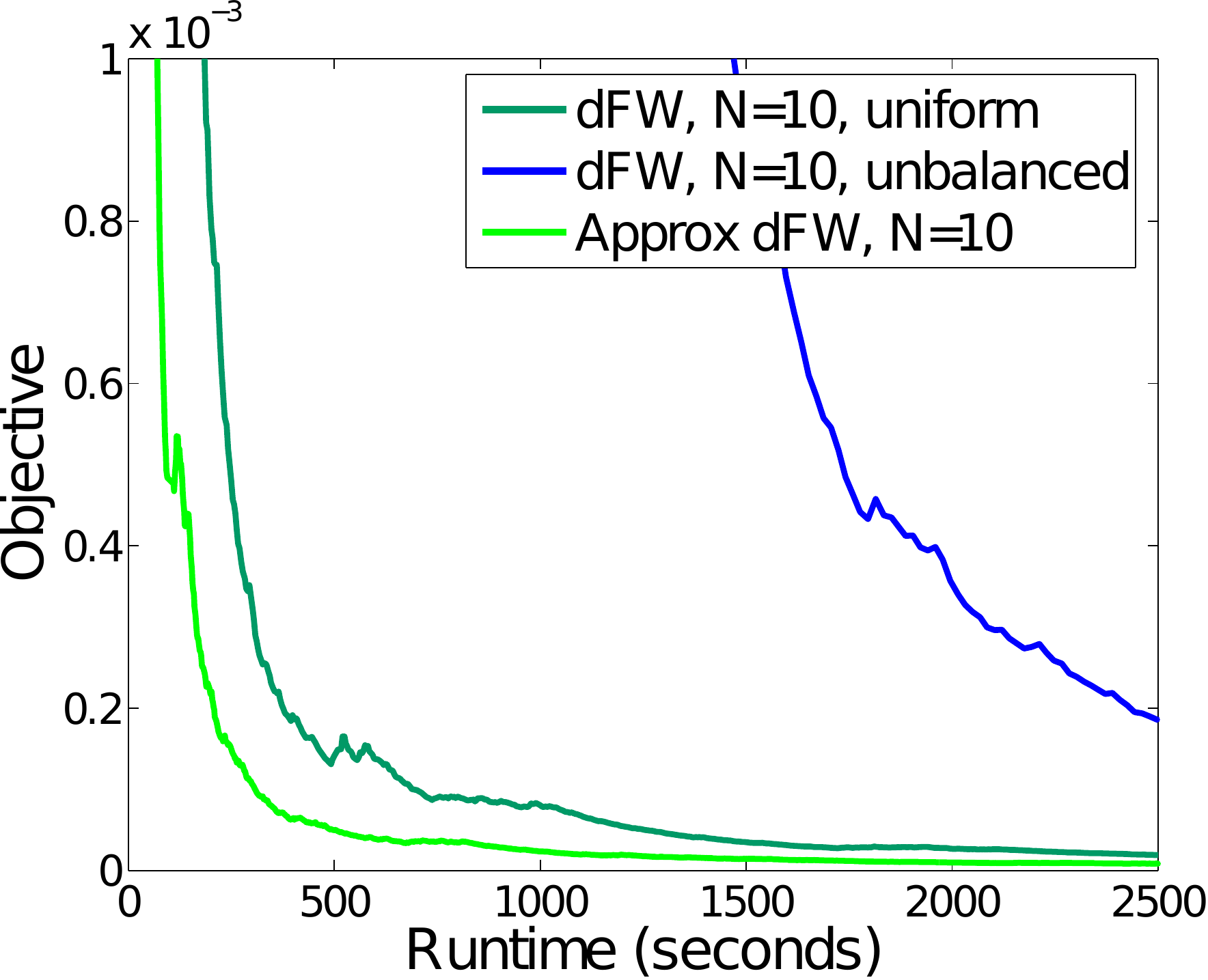}\hspace{-0.3cm}
\label{fig:mpi_approx}
}
\subfigure[dFW under asynchrony]{
\includegraphics[width=0.33\columnwidth]{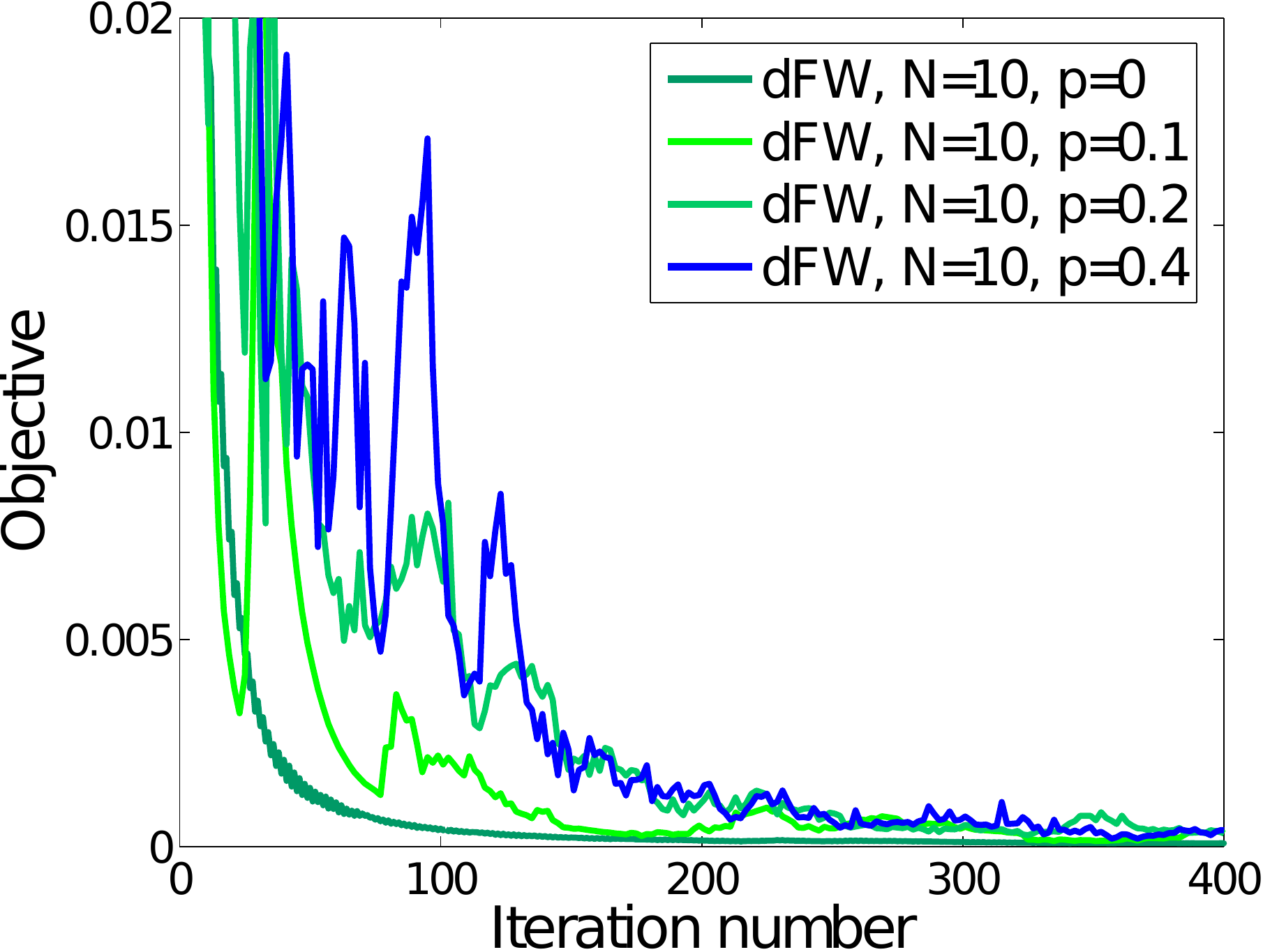}
\label{fig:dropcom}
}
\caption{Empirical evaluation of our theoretical analysis on distributed Kernel SVM (best seen in color).}
\vspace{-2mm}
\end{figure*}

\paragraph{Exact dFW} We first investigate how dFW scales with the number of nodes in the network. For this experiment, training points are assigned to nodes uniformly at random.
Figure~\ref{fig:mpi} shows the runtime results for different values of $N$. We obtain near-linear speedup, showing that synchronization costs are not a significant issue when the local computation is well-balanced across nodes.

\paragraph{Approximate variant} When the local costs are misaligned, the exact dFW will suffer from large wait time. As we have seen in Section~\ref{sec:approx}, we can use the approximate variant of dFW to balance these costs. To illustrate this, we set $N=10$ and assign about 50\% of the data to a single node while others node share the rest uniformly at random. We use the greedy clustering algorithm on the large node to generate a number of centers approximately equal to the number of atoms on the other nodes, thus balancing local computation.\footnote{For kernel SVM with RBF kernel, the difference between gradient entries of $f$ can be directly bounded by performing clustering in the augmented kernel space. This is because $|\nabla f(\vct{\alpha})_i - \nabla f(\vct{\alpha})_{j}| \leq 2\max_{\vct{z}_k} |\tilde{k}(\vct{z}_i,\vct{z}_k) - \tilde{k}(\vct{z}_{j},\vct{z}_k) | \leq \|\tilde{\varphi}(\vct{z}_i) - \tilde{\varphi}(\vct{z}_{j})\|^2$.} Figure~\ref{fig:mpi_approx} shows that the approximate dFW reduces the runtime dramatically over exact dFW on the unbalanced distribution. As often the case in large-scale datasets, training examples cluster well, so the additive error is negligible in this case.

\paragraph{Asynchronous updates} Another way to reduce waiting times and synchronization costs is to use asynchronous updates. To simulate this, we randomly drop the communication messages with some probability $p>0$. This means that at any given iteration: (i) the selected atom may be suboptimal, and (ii) some nodes may not update their iterate. As a result, the iterates across nodes are not synchronized anymore. Figure~\ref{fig:dropcom} shows the objective value averaged over all the nodes at each iteration. We can see that dFW is robust to this asynchronous setting and converges properly (albeit a bit slower), even with 40\% random communication drops. This suggests that an asynchronous version of dFW could lead to significant speedups in practice. A theoretical analysis of the algorithm in this challenging setting is an exciting direction for future work.

\section{Conclusion}
\label{sec:conclu}

We studied the problem of finding sparse combinations of distributed elements, focusing on minimizing the communication overhead. To this end, we proposed a distributed  Frank-Wolfe algorithm with favorable properties in communication cost. We showed a lower-bound to the communication cost for  the family of problems we consider. This established that the proposed dFW is optimal in its dependency on the approximation quality $\epsilon$. Experiments confirmed the practical utility of the approach. 

\paragraph{Acknowledgments} 
This research is partially supported by NSF grants CCF-0953192 and CCF-1101215, AFOSR grant FA9550-09-1-0538, ONR grant N00014-09-1-0751, a Google Research Award, a Microsoft Research Faculty Fellowship and the IARPA via DoD/ARL contract \# W911NF-12-C-0012. The U.S. Government is authorized to reproduce and distribute reprints for Governmental purposes notwithstanding any copyright annotation thereon. The views and conclusions contained herein are those of the authors and should not be interpreted as necessarily representing the official policies or endorsements, either expressed or implied, of IARPA, DoD/ARL, or the U.S. Government.

\bibliographystyle{plainnat}
\bibliography{arXiv_main}


\begin{appendices}

\section{Proof of Theorem~\ref{thm:dfw}}
\label{app:dfw}

We first prove the following claim.

\begin{claim}
\label{claim:exec}
At each iteration of dFW, each node has enough information to execute the algorithm. The communication cost per iteration is $O(Bd+NB)$, where $B$ is an upper bound on the cost of broadcasting a real number to all nodes in the network $G$.
\end{claim}
\begin{proof}
Let $k\geq 0$ and $v_i\in V$. We show that $v_i$ has enough information to compute its local gradient at iteration $k$. The $j^{\text{th}}$ entry of its local gradient is given by
\begin{equation*}
\label{eq:localgrad}
\nabla f(\vct{\alpha}^{(k)})_j = \left[\mat{A}\T\nabla g(\mat{A}\vct{\alpha}^{(k)})\right]_j = \vct{a}_j\T\nabla g(\mat{A}\vct{\alpha}^{(k)}),
\end{equation*}
where $j\in \mathcal{A}_i$. Notice that $\vct{\alpha}^{(k)}$ is a linear combination of $\{\vct{e}^{j^{(0)}},\vct{e}^{j^{(1)}},\dots,\vct{e}^{j^{(k-1)}}\}$, hence
$$\supp(\vct{\alpha}^{(k)}) = \left\{j^{(0)},j^{(1)},\dots,j^{(k-1)}\right\}.$$
Therefore, $\mat{A}\vct{\alpha}^{(k)}$ only depends on atoms $\{\vct{a}^{j^{(0)}},\vct{a}^{j^{(1)}},\dots,\vct{a}^{j^{(k-1)}}\}$ which were broadcasted at iterations $k'<k$ (step 4) and are thus available to $v_i$ at iteration $k$. Furthermore, atom $\vct{a}_j$ belongs to the local dataset of $v_i$ for all $j\in \mathcal{A}_i$. It follows that $v_i$ has enough information to compute its local gradient at iteration $k$ and thus to execute step 3. Moreover, it is easy to see that $v_i$ can execute steps 4, 5 and 6 based on information broadcasted at earlier steps of the same iteration, which completes the proof of the first part of the claim.

At iteration $k\geq 0$, dFW requires node $v_{i^{(k)}}$ to broadcast the atom $\vct{a}_{j_k}\in\mathbb{R}^d$, and every node $v_i\in V$ to broadcast a constant number of real values. Thus the communication cost of iteration $k$ is $O(Bd+NB)$.
\end{proof}

Using Claim~\ref{claim:exec}, we prove Theorem~\ref{thm:dfw}.

\begin{proof}
We prove this by showing that dFW is a Frank-Wolfe algorithm, in the sense that its updates are equivalent to those of Algorithm~\ref{alg:fw}. For $k\geq 0$, at iteration $k$ of dFW, we have:
\begin{eqnarray*}
j^{(k)} & = & \argmax_{i\in[N]} \left[\argmax_{j\in\mathcal{A}_i} \left| \nabla f(\vct{\alpha}^{(k)})_j\right|\right]\\
& = & \argmax_{i\in[N],j\in\mathcal{A}_i} \left| \nabla f(\vct{\alpha}^{(k)})_j\right|\\
& = & \argmax_{j\in[n]} \left| \nabla f(\vct{\alpha}^{(k)})_j\right|,
\end{eqnarray*}
where the last equality comes from the fact that $\bigcup_{i\in[N]} \mathcal{A}_i = [n]$. Moreover, $g^{(k)}_{i^{(k)}} = \nabla f(\vct{\alpha}^{(k)})_{j^{(k)}}$. Therefore, the update (step 5) of dFW is equivalent to the Frank-Wolfe update (Algorithm~\ref{alg:fw}) for the $\ell_1$ constraint (Algorithm~\ref{alg:local}). The stopping criterion is also equivalent since
$$\sum_{i=1}^N S^{(k)}_i = \sum_{i=1}^N \sum_{j\in\mathcal{A}_i} \alpha^{(k)}_j\nabla f(\vct{\alpha}^{(k)})_j = \innerp{\vct{\alpha}^{(k)},\nabla f(\vct{\alpha}^{(k)})},$$
and $\innerp{\vct{s}^{(k)},\nabla f(\vct{\alpha}^{(k)})} = -\beta\left|g^{(k)}_{i^{(k)}}\right|$.
Consequently, dFW terminates after $O(1/\epsilon)$ iterations and outputs a feasible $\tilde{\vct{\alpha}}$ that satisfies $f(\tilde{\vct{\alpha}}) - f(\vct{\alpha}^*) \leq h(\tilde{\vct{\alpha}}) \leq \epsilon$. Combining this result with Claim~\ref{claim:exec} establishes the total communication cost. The second part of the theorem follows directly from the fact that dFW uses Frank-Wolfe updates.
\end{proof}

\section{Extensions of the Lower Bound}
\label{app:bound}

In this section, we extend the communication lower bound (Theorem~\ref{thm:lower}) to two important cases.

\subsection{Problems with $\ell_1$ Constraint}

The lower bound also applies for the case when the constraint is $\|\vct{\alpha}\|_1\leq \beta$.
Choose $\beta=1$ and consider the following problem
$$\min_{\|\vct{\alpha}\|_1 \leq 1} \frac{\|\mat{A} \vct{\alpha}\|_2^2}{\|\vct{\alpha}\|_1^2}$$
where $\mat{A}$ is a $d\times d$ orthonormal matrix.
After scaling, it is equivalent to
$$\min_{\|\vct{\alpha}\|_1 = 1} \|\mat{A} \vct{\alpha} \|_2^2.$$
Since $ \|\mat{A} \vct{\alpha} \|_2^2 = \vct{\alpha}^\top \mat{A}^\top \mat{A} \vct{\alpha} = \|\vct{\alpha} \|_2^2$, taking absolute values of the components in $\vct{\alpha}$ does not change the objective value.
The problem is then equivalent to optimizing $\|\mat{A} \vct{\alpha}\|_2^2$ over the unit simplex, and the lower bound follows.

\subsection{Larger Number of Nodes}

The analysis also extends to the case where the network has more than two nodes.
Select two nodes that are farthest apart in the network,
put the data $\mat{A}_1$ and $\mat{A}_2$ on these two nodes,
and by the same argument,
the communication is lower bounded by $\Omega((\frac{d}{\epsilon} \log \kappa - d\log d)D)$
bits, where $D$ is the diameter of the graph.

\end{appendices}

\end{document}